\documentclass[11pt]{article}

\usepackage{graphicx}
\usepackage{amsmath, amsthm, amssymb, bbm}
\usepackage{amsfonts}
\usepackage{lipsum}
\usepackage{setspace}
\usepackage{natbib}
\usepackage{tikz, subcaption}
\setcitestyle{authoryear,open={(},close={)}} 

\newtheorem{theorem}{Theorem}
\newtheorem{claim}{Claim}

\newtheorem{corollary}{Corollary}

\title{\bf\Large Endogenous Network Structures \linebreak with Precision and Dimension Choices}
\author{Nikhil Kumar\footnote{Email: kumarnik@sas.upenn.edu. I am grateful to Kevin He for his guidance and many helpful discussions. I also thank Krishna Dasaratha, Rakesh Vohra, participants at the 2025 Midwest Theory Conference, and anonymous reviewers for their helpful comments and suggestions. All remaining errors and omissions are mine. }\\ 
\textit{University of Pennsylvania}}
\date{November 2025\vspace{-5ex}}

\begin{document}
\fontfamily{cmr}\selectfont
\maketitle
\vspace{0.2cm}

\normalsize \noindent {\bf Abstract:} This paper presents a social learning model using the DeGroot heuristic where agents form initial beliefs about a multi-dimensional state by choosing which dimension to learn about and exerting costly effort to acquire a signal. We then endogenize the network by assuming that agents’ initial information-acquisition choices affect how much weight they put on each other during the repeated belief updating. Agents not only choose the precision of their signals and what dimension of the state to learn about, but these decisions directly determine the underlying network structure on which social learning occurs. We show that under a fixed network structure, the optimal precision choice is sublinear in the agent's stationary influence in the network, and this individually optimal choice is smaller than the socially optimal choice by a factor of $n^{1/3}$ and a cost-function dependent ratio. The welfare gap grows quadratically in the number of agents. Under a dynamic network structure, we specify the network by defining a kernel distance between agents, determining how much weight agents place on one another. Agents choose dimensions to learn about such that their choice minimizes the squared sum of influences of all agents: a network with equally distributed influence across agents is ideal. 
\medskip
\renewcommand*{\thefootnote}{\arabic{footnote}}


\newpage

\section{Introduction}

Social learning allows groups to aggregate diverse information and learn efficiently on an underlying network. Learning agents face two intertwined choices: how much effort to invest in acquiring private information and what exactly to learn about. In many collaborative environments, agents then combine these private signals with social information. Relying on their peers’ opinions, agents specify weights on other opinions based on an underlying network structure, which then determines how social learning occurs. 

Information in networks is thus spread and aggregated through two main channels: private information and social information. Getting beneficial private information is costly, and putting in more effort to learn should translate to more precise information. Each agent can separately choose how much effort to exert, mapping into a precision choice and thus a level of private information. On the other hand, the spread of social information is directly specified by the overall network structure. The manner in which this information sharing occurs across a network, however, is very context-dependent. In particular, individual agent decisions can shape the underlying network structure. Information sharing is more likely between agents who choose to learn about similar things and so the underlying network structure on which social learning occurs should be dependent on agents' behavior. 

This paper presents a model that captures both of these phenomena under a DeGroot learning heuristic: $i)$ agent choices over both precision and which dimension of the state to learn about, $ii)$ an endogenous network structure dependent on agents' learning behavior. Most current literature in social learning and learning in networks focuses on a fixed network structure and an endowed signal to each agent: we consider relaxations of both. We utilize the DeGroot learning rule due to its tractability and intuitive interpretation, where agents update their opinions as a convex combination of their neighbors' opinions.

We first present a simple network learning structure and consider adaptations of the model in order of increasing complexity. The simplest model is one with a single-dimensional state, a static network structure, and a one-time learning decision. Agents then choose how much effort they put into learning, explicitly choosing the precision of the signal they receive and incurring higher costs for more precise signals. We show that the optimal choice of precision is sublinear in the agent's stationary influence in the network. In particular, under affine cost functions, the difference between the individually optimal and socially optimal precision choice is a factor of $n^{1/3}$. We also present examples of common network structures and the corresponding optimal precision choices on such networks. 


We then extend the underlying state to be multi-dimensional where agents have the ability to choose which dimension of the state they learn about. Agents receive imprecise but consistent signals about each dimension of the state, but can then choose a particular dimension to specialize in and learn more about. Under this structure, agents are indifferent on their dimension choice. We provide several applications of this model. Specifically, we introduce the notion of specialists and generalists, where agents have a tradeoff between the precision of their information and how much they can learn. Further, related to the multiplexing literature, we demonstrate that agents should choose to learn about the dimension on which they have the greatest stationary influence.

Finally, we consider the case in which the network structure is endogenously formed: agents first choose dimensions of the state to learn about, and based on those choices the network structure is formed. We posit that agents who choose to learn about similar dimensions place higher weights on each other's opinions, implying that network formation stems from a measure of informational  homophily between agents. Under this structure, we show that agents choose dimensions to learn about such that the squared sum of influence across agents is minimized. This result resembles the \lq\lq wisdom of the crowd" phenomenon but without the need for an overarching social planner. 

Learning on an endogenously formed network structure is further extended to the case of repeated interactions, where the network structure is based on similarities in agents' past dimension choices. Under a martingale assumption on dimension choice beliefs of other agents, the resulting optimal dimension choice is a direct application of the single iteration case. 

\subsection{Related Literature}

These results contribute to the large and growing research on social learning. The literature can be broken into two general threads: DeGroot learning and sequential social learning with Bayesian agents. This paper combines the DeGroot learning rule with ex-ante Bayesian maximization, where agents optimize knowing how their decisions will influence learning in the social network. DeGroot learning is a long standing area of research, where agents update their opinions as a linearly weighted sum of their neighbors' opinions \citep{degroot1974reaching}. This learning heuristic is well-discussed in the literature, with many later papers formalizing the informational environment under which DeGroot learning takes place (\citealp{demarzo2003persuasion}; \citealp{acemoglu2011opinion}; \citealp{golub2016learning}). 

Another thread of the social learning literature focuses on sequential learning, and a smaller subset of this literature focuses on agents choosing the precision of their own private signals. Specifically, \citet{mueller2016social} analyze social learning with costly search and show that asymptotic learning occurs as long as costs eventually approach zero. Further, \citet{ali2018herding} considers the question of costly information acquisition through a herding approach, showing that agents only choose to acquire information when they have a positive probability of changing the existing consensus. We formalize this notion of precision choice and allocated effort under the DeGroot learning heuristic.

\citet{sadler2021games} also discuss the notion of endogenous network formation. However, their model allows agents to choose both their connections and effort levels directly whereas in our framework agents choose their effort and the network structure is specified endogenously from these choices. We argue that explicitly choosing specific neighbors is often infeasible from an agent’s perspective; instead, we suggest that such choices can be captured indirectly through agents’ actions.

This paper also contributes to the new and recently growing literature on multiplexing in economic networks. \citet{chandrasekhar2024multiplexing} formally introduce the notion of multiplexing in networks and analyze diffusion with multiplex networks and an SIS model. \citet{candogan2025network} consider an extension of the disease spread model where agents can take actions on a given multiplexed network. We extend this strand of multiplexing literature and analyze the optimal agent choices when learning about an underlying state under a multiplexed network. We consider the case in which the underlying state is multi-dimensional, and the multiplexed network stems from different levels of information spreading across dimensions. 

Finally, our paper contributes to the more general literature on learning in dynamic environments and repeated interactions within such environments.  \citet{dasaratha2023learning} present a model with dynamic DeGroot learning, where Bayesian agents learn about a dynamically changing state. \citet{huang2024learning} analyze repeated interactions with long-lived
agents and show that the equilibrium speed of learning is upper bounded by the precision of the bounded signals. In this paper, we address similar core questions but consider a framework in which the network itself is constructed based on how agents choose to learn.

\section{The Model}
\label{model}

We analyze a social learning framework under DeGroot updating with costly information acquisition. Let $N = \{1, 2, ...n\}$ denote the finite set of agents. Agents are trying to learn the true value of a single-dimensional state $\theta$, and they learn and update under a  network $\mathcal{W}$. The network, or directed weight matrix, is common knowledge across all agents. Furthermore, assume that a stationary distribution exists under this network structure: i.e., there exists a distribution $ \pi \text{ such that } \pi \mathcal{W} = \pi, \sum_i \pi_i = 1$, and consequently that the network is strongly connected and aperiodic so that beliefs do indeed converge.\footnote{This assumption is not overly restrictive: imposing a uniform but very small positive lower bound $\epsilon > 0$ on all agent-pair influences implies the existence of a unique and positive stationary distribution.} 

\subsection{Precision Choice on a Fixed Network Structure}
\label{sec:precisionchoice}

We first consider the case under which the network structure $\mathcal{W}$ is fixed. Each agent then chooses a level of precision $\tau_i$ subject to a strictly increasing and convex cost function $c_i(\tau_i)$ to maximize the accuracy of the eventual network consensus. Each agent then receives a independent signal $s_i \sim \mathcal{S}$ dependent on their precision choice.  A higher level of effort corresponds to a more precise received signal, but effort is costly. We work with unbiased normal signals, where a choice of precision level $\tau_i$ implies agent $i$'s signal is a realized draw from $\mathcal{N}\left(\theta, \frac{1}{\tau_i^2}\right)$.

Each agent's utility function $u_i$ (or equivalently their loss function $\ell_i$) is a function of the network consensus accuracy and their individual cost incurred: $ -u_i(\tau_i)= \ell_i(\tau_i)= \mathbb{E}[(\hat{\theta}-\theta)^2] + c_i(\tau_i)$ where the expectation is with respect to all other agents' precision choices $\tau_1,..., \tau_{i-1}, \tau_{i+1}, ..., \tau_n$.\footnote{A more general utility function would allow an explicit trade-off of  network consensus error and individual cost functions. Let  $\alpha_i, \beta_i > 0$ denote the weights an agent places on consensus errors and individual costs respectively, yielding the loss function $\alpha_i \cdot \mathbb{E}[(\hat{\theta}-\theta)^2] + \beta_i \cdot c_i(\tau_i)$. However, this only adds a constant $\frac{\beta_i}{\alpha_i} > 0$ factor to the cost function, which a monotonic transformation and yields analogous analysis.}

Social learning occurs under the DeGroot updating heuristic: agents update their beliefs by taking a weighted average of their neighbors' beliefs according to the network $\mathcal{W}$. Under the existence of a stationary distribution of $\mathcal{W}$, the network consensus $\hat{\theta}$ can be determined solely by agents' initial signal realizations $s_i$ and the stationary distribution $\pi$ of $\mathcal{W}$: $\hat{\theta} = \pi^\top s = \sum_{k=1}^n \pi_k s_k$.

We then determine equilibrium precision choices for each agent by solving the corresponding minimization problems. Each agent chooses a precision level to minimize their loss, i.e. the sum of the network consensus error and their precision costs: 
$$\min_{\tau_i \ge0} \; \mathbb{E}[(\hat{\theta}-\theta)^2] + c_i(\tau_i)$$
The individual signal variances are a sufficient statistic for analyzing the network consensus, and this yields the following equilibrium result. 

\begin{theorem}[\text{Optimal Precision Choice Under Fixed Network Structure, One-Dimensional State}]
Under normal signals, each agent's optimal choice of precision is increasing in their influence in the network, but at a sublinear rate. Formally, each agent $i$ chooses an individually optimal precision choice of $$\tau_i = \biggl(\frac{2 \pi_i^2}{c_i'(\tau_i)} \biggr)^{1/3}$$

\label{sublinear}
\end{theorem}

\begin{proof}
    See Appendix \ref{thm1_DERIV} for details. Follows from independence and consistency of signals and first-order conditions.  
\end{proof}

Under an increasing cost function, this theorem illustrates that an agent should choose a higher precision if they have more influence in the network. Note that under linear costs, we have a direct proportionality of precision choice and network influence: $\tau_i^3 \propto  2 \pi_i^2 \Rightarrow \tau_i^* \propto 2 \pi_i^{2/3}$. 

Theorem \ref{sublinear} presents a dominant strategy for each agent in the network, but this behavior is not necessarily optimal from a social planner's point of view. Costs are effectively bundled together, and the social planner chooses a vector of precision choices for each agent to minimize the total objective value across all n agents:

\begin{equation}
    \min_{\tau_1, \tau_2, ... \tau_n \ge0} \; n \biggl[ \sum_{k=1}^n \pi_k^2 \cdot \mathrm{Var(s_k)} \biggr]+ \sum_i c_i(\tau_i) \label{socialplanner}
\end{equation}

\noindent which under normal signals, simplifies to the following expression:
$$\min_{\tau_1, \tau_2, ... \tau_n \ge0} \; n \biggl[\sum_{k=1}^n \frac{\pi_k^2}{\tau_k^2} \biggr]+ \sum_i c_i(\tau_i)$$

\begin{corollary}[Socially Optimal Precision Choice]
    Each agent's socially optimal precision choice is higher than the individually optimal choice in Theorem \ref{sublinear}. In particular, the ratio of the two must be bounded:
    $$1 \leq \frac{\tau_i^{social}}{\tau_i} \leq n^{1/3}$$ \label{corollary:socialplanner}
\end{corollary}
\begin{proof}[Proof Sketch:]
    Solve for first-order conditions and use convexity of costs. See details in Appendix \ref{pf:socialplanner}. 
\end{proof}

Under an affine cost function, the ratio of cost functions disappears and thus $\tau_i^{social} = \tau_i \cdot n^{1/3} $. Thus, the upper bound is reached under linear cost functions; the lower bound is trivially obtained with one agent. The gap between the two is a function of the number of agents and the steepness of the cost function, i.e. the elasticity of marginal costs. 

Under a cost function of a form $c_i(\tau_i) = \tau_i^{\alpha}$, the ratio of precision choices can be shown to be $\frac{\tau_i^{social}}{\tau_i} = n ^{\frac{1}{2 + \alpha}}$ (see Appendix \ref{pf:socialplanner} for details). A steeper cost function implies a smaller gap between precision choices, and we observe an instance of the free-rider problem that decays with the steepness of the cost function. Agents exert lower effort \textemdash \ choosing lower precision levels \textemdash \ believing that other agents will choose high enough precision levels to improve the network consensus. From a social planner's perspective, an agent exerting effort is directly beneficial for the whole network as it reduces consensus variance. Thus, specific agents may bear higher individual costs, but overall consensus variance will simultaneously fall for all agents leading to a better social outcome. Note that, however, this gap does not fully die out with a larger number of agents, and instead grows sub-linearly in $n$. 

Another interesting result following from Theorem \ref{sublinear} and Corollary \ref{corollary:socialplanner} is that the gap between the individually optimal and socially optimal outcomes is independent of the network structure. Optimal precision choices are increasing in an agent's network influence under both the individually and socially optimal outcome, but the difference between the two is dependent only on the number of agents and marginal cost elasticity. 

Our above results analyze the difference in precision choices; a natural question to ask is how this gap affects social welfare. 

\begin{claim}[Social Welfare Gap]
    The welfare gap between the individually optimal precision choices $\tau_i$ and the socially optimal precision choices $\tau_i^{social}$ is of order $\Theta(n^2)$. 
\end{claim}

\begin{proof}[Proof Sketch:]
    Define the welfare gap induced by each agent $i$ as $\delta_i$, which is the difference between their contribution to welfare under their two precision choices: 
    $$\delta_i = n \frac{\pi_i^2}{\tau_i^2} + c_i(\tau_i) - \Biggl(n \frac{\pi_i^2}{(\tau_i^{social})^2} + c_i(\tau_i^{social}) \Biggr) $$
    Use first-order conditions, convexity of costs, and Corollary \ref{corollary:socialplanner} to bound each individual's welfare gap, and summing over all $n$ agents gives the result above. See Appendix \ref{pf:welfareanalysis}.
\end{proof}

In particular, this result implies that individually optimal equilibria under social learning can be very problematic; the welfare gap between the individually optimal and socially optimal outcome grows at a quadratic rate in the number of agents. 

\subsubsection{Examples Across Standard Network Topologies}
\label{section:basicexamples}

Our results above indicate that the gap between optimal precision choice in the individually optimal and socially optimal case is independent of the network structure. We give some examples and discuss the implications below. 

First consider a fully connected network with $n = |\mathcal{N|}$ agents where each agent $i$ puts weight $x_i$ on themselves, and splits the remaining weight evenly on all other agents. The network is thus fully specified by a vector $x \in \mathbb{R}^n$, which specifies how much weight each agent puts on their own opinion. Figure \ref{fig:n-person-connected} illustrates the corresponding weight matrix and network structure for $n=8$. 

\begin{figure}[h!]
  \centering
  \begin{minipage}{0.45\textwidth}
  \Large
    \[
      W \;=\;
      \begin{pmatrix}
        x_1               & \tfrac{1 - x_1}{\,n-1\,} & \cdots & \tfrac{1 - x_1}{\,n-1\,} \vspace{0.2cm} \\
        \tfrac{1 - x_2}{\,n-1\,} & x_2               & \cdots & \tfrac{1 - x_2}{\,n-1\,} \\
        \vdots            & \vdots            & \ddots & \vdots \\
        \tfrac{1 - x_n}{\,n-1\,} & \tfrac{1 - x_n}{\,n-1\,} & \cdots & x_n
      \end{pmatrix}
    \]
  \end{minipage}
  \quad
  \normalsize
  \begin{minipage}{0.45\textwidth}
    \centering
    \begin{tikzpicture}[scale=1,
        node/.style={circle, draw, minimum size=8mm, fill=blue!20},
        edge/.style={-, thick},
        loopedge/.style={->, thick},
        lab/.style={fill=white, inner sep=1pt},
      ]
      \def\n{8}
      \def\r{2cm}
      \foreach \i in {1,...,\n}{
        \node[node] (N\i) at ({90 - 360*(\i-1)/\n}:\r) {\(\i\)};
      }
      \foreach \i in {1,...,\n}{
        \foreach \j in {1,...,\n}{
          \ifnum\i<\j
            \draw[edge] (N\i) -- (N\j);
          \fi
        }
      }
      \draw[loopedge] (N1) to[loop above]   node[lab] {\(x_{1}\)} (N1);
      \draw[loopedge] (N2) to[loop above]   node[lab] {\(x_{2}\)} (N2);
      \draw[loopedge] (N3) to[loop right]    node[lab] {\(x_{3}\)} (N3);
      \draw[loopedge] (N4) to[loop right]    node[lab] {\(x_{4}\)} (N4);
      \draw[loopedge] (N5) to[loop below]   node[lab] {\(x_{5}\)} (N5);
      \draw[loopedge] (N6) to[loop left]   node[lab] {\(x_{6}\)} (N6);
      \draw[loopedge] (N7) to[loop left]    node[lab] {\(x_{7}\)} (N7);
      \draw[loopedge] (N8) to[loop above]   node[lab] {\(x_{8}\)} (N8);
    \end{tikzpicture}
  \end{minipage}
  \caption{A general $n$-agent complete network weight matrix parameterized by self‐weights $x_i$ and a corresponding simplified network diagram for $n=8$.}
  \label{fig:n-person-connected}
\end{figure}

The stationary distribution and precision choice (from Theorem \ref{sublinear}) under this network structure for a general $n$ are:
$$\pi_i = \frac{\frac{1}{1-x_i}}{\sum\limits_{k=1}^n \frac{1}{1-x_k}} \text{ and } \tau_i = \left(\frac{2\pi_i^2}{c_i'(\tau_i)}\right)^\frac{1}{3}$$
and thus, the precision choice is strictly increasing in an agent's self-weight $x_i$. An agent with stronger influence in the network is thus willing to exert more effort to improve the network consensus. Details are given in Appendix \ref{proof:claimopt_prec}. 

A complete network structure yields a fully symmetric stationary distribution. We can also consider asymmetric network structures with a central highly influential agent. Figure \ref{fig:cp-star-networks} presents two examples of such networks: a core-periphery network with a single core agent, and a star network. The core-periphery network presents a network under which periphery agents are connected to one another and equally weight themselves, their two adjacent neighbors, and the central agent. However, in the star network, the only mode of information diffusion is via the central agent. In both cases, the central agent can be viewed as an expert, and she naturally has a persistent level of influence:
$$\pi_n^{cp} = \frac{n}{5n-4} \text{ and } \pi_n^{star} = \frac{n}{3n-2}$$
The exact details and complete stationary distributions can be found in Appendix \ref{proof:core-perip-claim}.

\begin{figure}[h!]
  \centering
  \begin{minipage}{0.45\textwidth}
    \centering
    \begin{tikzpicture}[
        periph/.style={circle,draw,minimum width=8mm,minimum height=8mm,fill=blue!30,inner sep=1pt,font=\small},
        core/.style={circle,draw,minimum width=8mm,minimum height=8mm,fill=red!30,inner sep=1pt,font=\small},
        interedge/.style={-,thick},
        loopedge/.style={->,thick},
        lab/.style={fill=white,inner sep=1pt,font=\tiny},
      ]
      \node[periph] (1)  at (90:2)   {1};
      \node[periph] (2)  at (30:2)   {2};
      \node[periph] (3)  at (-30:2)  {3};
      \node[periph] (n3) at (-90:2)  {\(n\!-\!3\)};
      \node[periph] (n2) at (-150:2) {\(n\!-\!2\)};
      \node[periph] (n1) at (150:2)  {\(n\!-\!1\)};
      \node[core]   (n)  at (0,0)    {\(n\)};
      \draw[interedge] (1)--(2);
      \draw[interedge] (2)--(3);
      \draw[interedge] (n3)--(n2);
      \draw[interedge] (n2)--(n1);
      \draw[interedge] (n1)--(1);
      \draw[interedge,loosely dotted,thick] (3)--(n3);
      \foreach \u in {1,2,3,n3,n2,n1}{
        \draw[interedge] (\u)--(n);
      }
      \foreach \u/\pos in {1/above,2/right,3/right,n3/below,n2/left,n1/left}{
        \draw[loopedge] (\u) to[loop \pos] node[lab] {1/4} (\u);
      }
      \draw[loopedge] (n) to[out=10,in=60,looseness=7] node[lab] {\(\tfrac1n\)} (n);
    \end{tikzpicture}
    \subcaption{Core–periphery network}
  \end{minipage}
  \hspace{0.05\textwidth}
  \begin{minipage}{0.45\textwidth}
    \centering
    \begin{tikzpicture}[
        periph/.style={circle,draw,minimum size=8mm,fill=blue!30,inner sep=1pt,font=\small},
        core/.style={circle,draw,minimum size=8mm,fill=red!30,inner sep=1pt,font=\small},
        interedge/.style={-,thick},
        loopedge/.style={->,thick},
        lab/.style={fill=white,inner sep=1pt,font=\tiny},
      ]
      \node[periph] (1)   at (90:2)   {1};
      \node[periph] (2)   at (30:2)   {2};
      \node[periph] (3)   at (-30:2)  {3};
      \node[periph] (n3)  at (-90:2)  {\(n\!-\!3\)};
      \node[periph] (n2)  at (-150:2) {\(n\!-\!2\)};
      \node[periph] (n1)  at (150:2)  {\(n\!-\!1\)};
      \node[core]   (n)   at (0,0)    {\(n\)};
      \foreach \u in {1,2,3,n3,n2,n1}{
        \draw[interedge] (\u)--(n);
      }
      \foreach \u/\pos in {1/above,2/right,3/right,n3/below,n2/left,n1/left}{
        \draw[loopedge] (\u) to[loop \pos] node[lab] {1/2} (\u);
      }
      \draw[loopedge] (n) to[out=10,in=60,looseness=7] node[lab] {\(\tfrac1n\)} (n);
    \end{tikzpicture}
    \subcaption{Star network}
  \end{minipage}
  \caption{Comparison of the two network topologies: (a) core–periphery structure and (b) star structure, each with $n$ agents.}
  \label{fig:cp-star-networks}
\end{figure}

In such networks, we recover the expected intuition: an agent that has more influence in a social network should exert more effort to improve the network consensus. Note that the central agent is also interested in learning the true state; an interesting question to ask is what happens if the central agent is adversarial. We discuss related questions in Section \ref{conclusion}.

\subsection{Multi-Dimension Choice on a Fixed Network Structure}
\label{multi-dim:sec}

We now assume that the network structure is fixed and all agents have the same precision $\tau$ but extend the state to higher dimensions. The agent's choice now becomes a choice of a single element of the state to learn about and correspondingly receive a signal on. Since there are multiple dimensions of the state, the DeGroot updating is done independently on each dimension.

The true multi-dimensional state $\theta \in \mathbb{R}^m$, and so the network consensus as a result of DeGroot updating $\hat{\theta} \in \mathbb{R}^m$ is constructed in a similar manner to above. Rather than each agent's initial opinion being a scalar (which in 2.1 is equivalent to their signal realization), each agent now has an opinion/belief vector which they update according to their neighbors. We implicitly assume that the number of agents in the network is much larger than the dimension of the state, meaning that learning occurs sufficiently for each dimension.

To maintain consistency of the overall network consensus, rather than assuming an improper prior where agents have a prior of $\theta_j = 0$ for dimensions they do not learn about, we assume every agent learns about all other dimensions but just with very low precision. This resembles real-world phenomena: an individual may just see a single advertisement or social media post about an underlying state (company valuation, news, etc) and form their belief solely based on that noisy piece of information. In particular, if agent i chooses to learn about dimension $d_i$, for all other dimensions $d_j$, their estimate of that dimension will be a sample from a very high variance distribution around the true state dimension $\theta_j$.\footnote{This ensures that the network estimate on every dimension is consistent in asymptotics. For example, repeated DeGroot updating on a dimension with one unbiased signals and others with biased priors will lead to an entirely incorrect consensus with unbounded error.}

All agents share the same precision $\tau > 0$, and the network weight matrix $W$ has a stationary distribution $\pi$.  Each agent $i$ then chooses one coordinate $d_{i}\in\{1,\dots,m\}$ on which to sample a signal:
\[
s_{i,d_{i}}\;\sim\;\mathcal N\!\bigl(\theta_{d_i},\,1/\tau^2\bigr),
\]
and for all other dimensions $j\neq d_i$, the agent receives a very noisy but consistent signal about the state: $\forall d_j \neq d_i, s_{i, d_j} \sim \mathcal{N}\left(\theta_{d_j}, 1/\underline{\tau}^2\right)$ where $\underline{\tau} << \tau$. 
Therefore, each agent has an initial estimate of the state of the form $e_{i} = [s_{i,1}, s_{i, 2}, ... s_{i, d_i}, ..., s_{i, m}]^\top$, which captures agent $i$ choosing to learn about only dimension $d_i$ of the state and receiving very noisy signals about all other dimensions. Their optimization problem now stems from choosing an optimal decision to learn about rather than a choice of signal precision.

If the network structure is uniform across state dimensions (i.e. the network structure on which DeGroot learning occurs is identical for all state dimensions), agents are trivially indifferent on which dimension of the state they learn about. No agent will have an incentive to specifically choose a certain dimension to learn about a priori. Furthermore, the network consensus is unbiased by a standard law of large numbers argument: $\hat{\theta}_j \rightarrow \theta_j \; \forall j \in \{1, ..., m\}$ as $n \rightarrow \infty$. This implies that all agents choosing the same dimension is indeed a dominant strategy, as is a uniformly random choice of dimension to learn about. 

\subsubsection{Multiplexed Network Structures}

We now briefly consider the case in which networks are distinct on different dimensions of the underlying state. Similar to the notion of multiplexing in networks introduced by \citet{chandrasekhar2024multiplexing}, two agents can be connected on many different layers and to different degrees. 

Formally, define a layer as a dimension of the underlying state agents are learning about. Each layer then corresponds to a distinct network structure, and thus a distinct stationary distribution. Let $\pi_i^j$ denote agent $i$'s stationary influence on the network layer corresponding to dimension $j$. An agent wants to choose a dimension to learn about based on the fixed but distinct network structures on different dimensions. 

Under such multiplexed networks, we claim that agents should choose to learn about the dimension on which they have the highest stationary influence:
$$d_i \in \arg \min_d \sum\limits_{j=1}^m \mathbbm{1} \{d = j\} \; \cdot  \pi_i^j$$

This implies that an agent chooses to learn more about dimensions on which they are regarded as experts. In particular, experts are agents whose opinions are valued more by the rest of the network, meaning that they have higher stationary influence. Even when all agents have identical precision levels, though, an agent should still choose to learn about the dimension on which they have the highest influence, even if they possess no comparative advantage over other agents.

Consider the following illustrative example: let $m=3$ and $n \geq 4$. 
Assume the three state dimensions have the three example networks discussed in Section \ref{section:basicexamples}: $d_1$: complete network, $d_2$: core-periphery/ring network, $d_3$: star network.\footnote{Without loss of generality, denote the central agent on dimension 2 and 3 to be agents n and 1 respectively.} Each agent will choose the dimension on which they have the highest corresponding stationary influence, and thus:
$$d_1 = 3, d_n = 2, \text{ and } d_i = 1 \;\;\forall i \in \{2, ..., n-1\}$$

\begin{figure}[!htbp]
  \centering
  \tikzset{
    baseNode/.style={circle, draw, minimum size=8mm, fill=blue!20, font=\small, inner sep=0pt},
    chosenNode/.style={circle, draw=black, ultra thick, minimum size=8mm, fill=yellow!60, font=\small, inner sep=0pt},
    edge/.style={-, thick},
  }
  \def\n{8} 
  \def\r{2.0cm}

  \begin{subfigure}[b]{0.32\textwidth}
    \centering
    \begin{tikzpicture}[scale=1]
      \foreach \i/\lbl in {1/{1},2/{2},3/{3},4/{n\!-\!3},5/{n\!-\!2},6/{n\!-\!1},7/{n}}{
        \ifnum\i=1
          \node[baseNode] (C\i) at ({90 - 360*(\i-1)/7}:\r) {\(\lbl\)};
        \else
          \ifnum\i=7
            \node[baseNode] (C\i) at ({90 - 360*(\i-1)/7}:\r) {\(\lbl\)};
          \else
            \node[chosenNode] (C\i) at ({90 - 360*(\i-1)/7}:\r) {\(\lbl\)};
          \fi
        \fi
      }
      \foreach \i in {1,...,7}{
        \foreach \j in {1,...,7}{
          \ifnum\i<\j
            \draw[edge] (C\i) -- (C\j);
          \fi
        }
      }
    \end{tikzpicture}
    \vspace{2mm}
    \caption*{Dimension 1: Complete network}
  \end{subfigure}
  \hfill
  \begin{subfigure}[b]{0.32\textwidth}
    \centering
    \begin{tikzpicture}[scale=1]
      \node[chosenNode] (Kc) at (0,0) {\(n\)};
      \foreach \i/\lbl in {1/{1},2/{2},3/{3},4/{n\!-\!3},5/{n\!-\!2},6/{n\!-\!1}}{
        \node[baseNode] (K\i) at ({90 - 360*(\i-1)/6}:\r) {\(\lbl\)};
        \draw[edge] (K\i) -- (Kc);
      }
      \foreach \i in {1,...,5}{
        \pgfmathtruncatemacro{\j}{\i+1}
        \draw[edge] (K\i) -- (K\j);
      }
      \draw[edge] (K1) -- (K6);
    \end{tikzpicture}
    \vspace{2mm}
    \caption{Dimension 2: Core–periphery}
  \end{subfigure}
  \hfill
  \begin{subfigure}[b]{0.32\textwidth}
    \centering
    \begin{tikzpicture}[scale=1]
      \node[chosenNode] (Sc) at (0,0) {\(1\)};
      \foreach \i/\lbl in {1/{2},2/{3},3/{n\!-\!3},4/{n\!-\!2},5/{n\!-\!1},6/{n}}{
        \node[baseNode] (S\i) at ({90 - 360*(\i-1)/6}:\r) {\(\lbl\)};
        \draw[edge] (S\i) -- (Sc);
      }
    \end{tikzpicture}
    \vspace{2mm}
    \caption{Dimension 3: Star network}
  \end{subfigure}

  \caption{Dimension choices to learn about: agent $n$ picks dimension 2, and agent $1$ picks dimension 3, and all others choose dimension 1. Yellow nodes indicate the agents who choose that corresponding dimension.}
  \label{fig:three-dim-choice-final}
\end{figure}

In such networks, different dimensions may be viewed as different areas of expertise, and so agent $i$ may have a very strong influence on the first dimension (i.e. high $\pi_i^1$) but may be much less influential on the second dimension (small $\pi_i^2$).

\subsection{Varying Precision Choices}
\label{varying-prec-multidim}

\subsubsection{Single Precision Parameter}

We now analyze the case in which precision parameters are not fixed across agents: i.e. the choice of each agent is a joint optimization problem over $\tau$ and $d$. Combining the choices from Section \ref{sec:precisionchoice} and \ref{multi-dim:sec}, we see that the joint optimization is decoupled. Since the network structure is such that agents share their entire opinion vector with their neighbors, the assumption of a consistent prior on every dimension and a fixed network structure implies agents have no reason to choose a certain dimension to learn about over another. Therefore, the choices of precisions would be identical to that of Theorem \ref{sublinear}. 

\subsubsection{Effort Allocation Across Dimensions}

The more interesting case of this joint optimization is when agents are not constrained to choosing a single dimension to learn about, but rather they can choose to allocate their effort across different dimensions. When investing effort into learning about the underlying state $\theta$, agents may thus choose to diversify their efforts across different dimensions. Rather than putting all their effort about some dimension $d_i$, for example, an agent may choose to put a uniform amount of effort on all dimensions of the state and thus improve the precision of their signals on every dimension. We call such an agent a \textit{generalist}, and call an agent that only allocates effort on one dimension of the state a \textit{specialist}.

As in Section \ref{multi-dim:sec}, every agent's initial opinion vector consists of scalar estimates for each dimension of the state. In Section \ref{multi-dim:sec}, however, the agent only chooses a single dimension to receive a strong signal, whereas they receive a very noisy signal for all other dimensions. We consider the case in each agent has a budget to spend on precision. We focus on the world with only specialists and generalists.\footnote{More intermediate agents can be included, but an agent who is a hybrid of a specialist and generalist can be represented by a convex combination of the two extremes. Our simplification thus improves tractability, as otherwise each agent would be solving an $\mathbb{R}^{m-1}$ dimension problem.}

Suppose the number of agents is linearly increasing in the number of state dimensions: let $n = cm$, where $c \in \mathbb{Z}^{+}$. Half of the agents are generalists and the other half are specialists. Following Section \ref{multi-dim:sec}, specialists will uniformly choose a dimension of the state to learn about, so on average, each state will have $\frac{c}{2}$ specialists. The generalists will split their effort equally on all m dimensions. We consider three population distributions in Figure \ref{fig:spec-vs-gen}.

\begin{figure}[h!]
  \centering
  {\small
  \begin{tabular}{cc|c|c}
    \# of Specialists & \# of Generalists 
      & \textbf{Specialist Signal Variances} 
      & \textbf{Generalist Signal Variances} \\
    \hline
    \noalign{\vskip 2mm} 
    \(0\)   & \(cm\) 
      & N/A 
      & \(\displaystyle \frac{1}{(\underline\tau + \tau_i/m)^2}\quad\forall\,j\) \\[8pt]
    \(\dfrac{cm}{2}\)  & \(\dfrac{cm}{2}\)
      & \(\displaystyle 
        \begin{cases}
          1/\tau_i^2 & \text{on chosen dimension},\\
          1/\underline\tau^2 & \text{on other $m-1$ dimensions}
        \end{cases}
      \) 
      & \(\displaystyle \frac{1}{(\underline\tau + \tau_i/m)^2}\quad\forall\,j\) \\[18pt]
    \(cm\)  & \(0\)  
      & \(\displaystyle 
        \begin{cases}
          1/\tau_i^2 & \text{on chosen dimension},\\
          1/\underline\tau^2 & \text{on other $m-1$ dimensions}
        \end{cases}
      \) 
      & N/A \\
  \end{tabular}
  }
  \caption{Different proportions of specialists and generalists in the network. Specialists concentrate effort on one coordinate (variance \(1/\tau_i^2\) there, \(1/\underline\tau^2\) elsewhere); generalists split effort equally (variance \(1/(\underline\tau+\tau_i/m)^2\) on every coordinate).}
  \label{fig:spec-vs-gen}
\end{figure}

Each agent's allocation budget $\tau_i$ is entirely determined by their influence in the network (by Theorem \ref{sublinear}). Therefore, the difference between specialists and generalists is just how this effective budget of $\tau_i$ is allocated. The proportion of specialists in the network is captured by $\alpha$, and thus the overall population of agents consists of $\alpha n $ specialists and $(1-\alpha)n$ generalists. 

Consider a complete network structure in which each agent has equal influence. The stationary distribution influence for each agent is thus $\frac{1}{n}$. We can then compute the overall consensus variance by summing up variances on each state dimension. The $\alpha n$ specialists each choose dimensions to learn about uniformly at random, and thus the expected number of specialists who learn about a specific dimension $d_i$ is $\frac{\alpha n}{m}$. Each of the $(1-\alpha)n$ generalists will learn a little bit about every dimension. 

\begin{claim}
    Under a complete network structure with equal influence by each agent, there is no interior optimal proportion of specialists. In particular, if 
    $$ \frac{1}{m \tau_i^2} + \frac{(m-1)}{m \underline{\tau}^2} - \frac{1}{\left(\underline{\tau} + \frac{\tau_i}{m}\right)^2} < 0$$ then $\alpha^* = 1$ and a network with only specialists is optimal. If the expression above is negative, $\alpha^* = 0$.
    \label{claim:spec-gen-comparison}
\end{claim}

When baseline signals are very noisy, the overall network is better off with generalists than specialists as learning a bit about everything lowers variance more than multiple specialized learners. Very noisy baseline signals means that $\underline{\tau} << \tau_i$, and thus $\alpha^* = 0$. Details are provided in Appendix \ref{examples_spec-vs-gen:app}. The exact condition in Claim \ref{claim:spec-gen-comparison} depends on the network structure: under a complete network structure, the stationary distribution is $\frac{1}{n}$ for all agents, and so the $\pi_k$ terms factor out of the first order condition. For remaining analysis in future sections, we continue with the assumption that all agents are specialists (i.e. only choose one dimension to learn about) but discuss extensions in Section \ref{conclusion}.

\section{Dynamic Network Structures}

We now consider the richer case of our model in which agents' choices endogenously determine the network structure. In particular, all agents are assumed to have the same precision, but the state is multi-dimensional and the network structure is flexible. We draw inspiration from \citet{sadler2021games} who consider endogenous networks formed by agents themselves choosing their set of neighbors. However, we argue that the choice and outcome variable should be reversed: agents choose how to learn and this in turn endogenously specifies the network structure. 

We argue this alternative choice is more in agreement with the canonical work of \citet{simon1955behavioral} and \citet{kahneman2003maps} on bounded rationality. Having the freedom to choose one's neighbors directly can often be overwhelming for an agent (as discussed by \cite{puri2018preference}, \cite{gabaix2024complexity}, and many others), and thus we argue that one's neighbors are implicitly characterized by learning behavior. Furthermore, agents are often not directly able to choose their neighbors in actuality; rather, their actions induce a set of neighbors. 

Therefore, we present a framework under which agents can choose their private learning behavior, and these choices implicitly characterize the underlying network structure on which social learning occurs. For example, two agents who choose to allocate effort learning about the same dimension are more likely to interact and overlap, and thus the corresponding network structure gives a higher weight to their shared edge. We first analyze the case in which a single iteration of social learning occurs: agents choose what to learn about, these choices induce a network structure, and learning occurs with the DeGroot heuristic until convergence. We then allow for social learning to occur in iterations, where an agent's sequence of learning behavior then induces the network structure in subsequent periods.  

\subsection{Single DeGroot Learning Iteration}
Agents' neighbors are determined by what element of the state they chose to learn about. Choices are simultaneously made by all agents, and the resulting choices determine the network structure. If two agents chose to learn about the same element of the state, they are more likely to be connected. 

In particular, $\mathcal{W}$ is dependent on agent's choices of dimensions $d_j$ they learn about. We assume that $ \forall i, j \; \mathcal{W}_{ij} > 0$. For an arbitrary agent $i$, her choice of dimension $d_j$ to learn about directly characterizes how she weights other people's opinions in the network. As in Section \ref{multi-dim:sec}, each agent's estimate is a $m$-dimensional vector, where all dimension estimate are unbiased but very imprecise on every dimension except the one chosen by the agent to learn about. 

Let $d = (d_1, d_2, ... d_n)$ represent the $n$ agents' choices of dimensions. To quantify the mapping from an agent's learning behavior to the induced network structure, we define the symmetric weight kernel $\mathcal{K}: \mathcal{N} \times \mathcal{N} \rightarrow \mathbb{R}$, which quantifies how strongly influenced two agents are by one another. We use a standard Gaussian kernel:
\begin{equation}
    \mathcal{K}(d_i, d_j) = \exp\left(-\alpha(d_i - d_j)^2\right) \label{kernel}
\end{equation}
where $\alpha$ captures the spread of this weight distribution. The network $\mathcal{W}$ is thus defined as the normalized kernel:
\begin{equation}
    \mathcal{W}_{ij} =\frac{\mathcal{K}(d_i,d_j)}{\sum\limits_{k=1}^n \mathcal{K}(d_i,d_k)}
\end{equation}

Our specification of the network kernel is inspired by the phenomenon of echo chambers \citep{nguyen2020echo}. Our kernel-based structure leads agents to place more weight on others whose chosen dimensions are closer (in squared distance) to their own. By imposing that agents place larger weights on neighbors whose learning behavior matches their own, we capture how homophily can drive self-reinforcing information loops. People interact more with others similar to them, and thus the information they use to construct their beliefs is self-enforcing. We discuss other potential network formation structures in Section \ref{conclusion}. Note that the agent puts the highest weight on her own opinion/others who learned about the same dimension as them: $\mathcal{K}(d_i, d_j)$ is maximized when $j = i$. 

Assuming that chosen precision is constant across agents, this formulation implies that the resulting stationary distribution $\pi$ is now a function of agent choices. Each agent's optimization problem is more complex: 
$$\min_{d_i} \; \mathbb{E}\biggl[\Bigl(\hat{\theta}(d_i, d_{-i})-\theta\Bigr)^2\biggr] + c_i(\tau)$$

As in Section \ref{model}, since signals are consistent across all dimensions, the first term is equivalent to the variance of the network consensus. Each agent wants to choose a dimension $d_i$ to reduce overall variance as much as possible. 

\begin{theorem}[\text{Dimension Choices under Multi-Dimensional State}]
    Each agent chooses to learn about the dimension $d_i$ which best distributes squared influence across agents in the network. In other words, they choose a dimension $$d_i \in \arg\min_{d} \sum\limits_{k=1}^n \left(\pi_k(d, d_{-i})\right)^2$$ \label{dimchoice:thm}
\end{theorem}

\begin{proof}[Proof:]

The overall variance of the consensus estimate is: $$\sum\limits_{j=1}^m \mathrm{Var(}\hat{\theta}_j) = \sum\limits_{j=1}^m \mathrm{Var}\left(\sum\limits_{i=1}^n \mathbbm{1} \{d_i = j\} \; \pi(d) s_{i, d_i}\right) = \frac{1}{\tau^2}\sum\limits_{j=1}^m \sum\limits_{i: d_i = j} \left(\pi(d_i, d_{-i})\right)^2$$ However, note that since each agent only learns about one dimension in each period, the summation is just the squared sum of each agent's stationary influence. Thus, each agent chooses $d_i$ to minimize:
$$\frac{1}{\tau^2}\sum\limits_{k=1}^n \left(\pi_k(d_i, d_{-i})\right)^2 + c_i(\tau)$$
The cost term is a constant under the optimization problem, and thus the theorem follows.
\end{proof}

Our result implies that choosing the dimension such that she has the least influence is not necessarily optimal for an agent. Rather, the agent wants to choose a dimension so that her influence is closest to $\frac{1}{n}$. Intuitively, the theorem claims that an agent should choose their dimension in a way that best distributes agents across state dimensions. The effect of a choice $d_i$ will affect the overall objective in two counteracting ways: $i)$ her own stationary influence, $ii)$ other agents' stationary influence as a result of agent $i$'s choice. 

This result strongly resembles the classic wisdom of crowds phenomenon (\citealp{golub2010naive}), which shows that with a social planner, evenly balanced influence is best for social learning. However, the key distinction in Theorem \ref{dimchoice:thm} is that individual agents independently choose their sampling dimension and achieve the same result. The result of balanced influence above is based on individual choice; since the objective function is dependent only on overall squared consensus error. 

For illustration, suppose there is a scenario in which she can choose $d_i = 1 \rightarrow \pi_1(d) = \frac{1}{n^2}$ or choose $d_i = 2 \rightarrow \pi_1(d) = \frac{1}{n}$. A choice of $d_i = 1$ minimizes her individual influence in the network, but that \lq\lq saved\rq\rq  \hspace{0.1cm} influence has to effectively be allocated to other agents, leading to higher overall network variance compared to the choice $d_i = 2$. 

The parameter $\alpha$ captures the spread of the weight kernel. If $\alpha \rightarrow 0$, then $\mathcal{K}(d_i, d_j) \rightarrow \frac{1}{n}$ which essentially collapsed to the case in Section \ref{model} where agents choose dimensions uniformly. The choice of dimension has no effect on the network structure, and so every dimension is effectively a best response. As $\alpha \rightarrow \infty$, $K(d_i, d_j) \scriptsize\rightarrow \begin{cases}
    \frac{1}{N_{d_i}} \;\text{ if $d_i = d_j$} \\
    0 \;\;\;\;\;\text{ otherwise}
\end{cases}$\normalsize
and so agent $i$ only assigns positive weight to other agents who choose the same dimension $d_i$ as them.

Rather than learning in a simultaneous fashion, we can also consider the case in which choices are made sequentially. Each agent observes the dimension choices of the agents who chose before them, and then they themselves choose a dimension to learn about. 

This specification resembles existing literature on sequential learning, but the main distinction is that the network structure is not determined until agents make their choices. In particular, an agent chooses a dimension such that the resulting endogenously formed network structure under their belief is best to learn about the underlying state $\theta$. We discuss this in Section \ref{conclusion} as a potential direction for future work. 

\subsection{Iterative DeGroot Learning}

We now consider the case in which this DeGroot updating process happens iteratively. In particular, in the first period, all agents choose a dimension of the state to learn about. These choices endogenously determine the network structure, and thus the corresponding DeGroot updating weights. DeGroot updating is simulated for a fixed number of periods (until the network arrives at or close to consensus), and then agents get a new opportunity to acquire information. However, the new network structure (DeGroot weights) are endogenously determined not only by the choices of what to learn about in the second time period, but also the first time period. We repeat this process and capture agents' learning behavior as sequences of chosen dimensions. 

Under this iterative framework, we can view learning behavior sequences as an agent's memory. \cite{malmendier2021memory}, \cite{fudenberg2024selective}, and many others present models for individual decision making and equilibria with an emphasis on memory. We draw inspiration upon these models, proposing that memory not only impacts how agents learn in future periods, but it also specifies the network under which learning occurs. 

We extend the kernel in Equation \ref{kernel} to compare not just choices of dimension, but rather players' histories of dimension choices, which we refer to as each player's \textit{internal memory}. In particular, consider two arbitrary players $i, i'$ with histories of dimension choices $\mathcal{M}_i^t = (d_i^1, d_i^2, ..., d_i^t)$ and $\mathcal{M}_{i'} = (d_{i'}^1, d_{i'}^2, ..., d_{i'}^t)$. $\mathcal{M}_i^t$ is agent $i$'s memory at time $t$ of all her past dimension choices. In Equation \ref{kernel}, choices are just scalars and so the squared distance metric is natural. However, with vector comparisons between two agents' memories, we propose the following distance metric:
\begin{equation}
    \mathcal{D}(\mathcal{M}_i^t, \mathcal{M}_j^t) = \sum\limits_{\tau=1}^t \gamma^{t-\tau} \left(d_i^\tau - d_j^\tau \right)^2 \label{eq:distance-memories}
\end{equation}
which is effectively an exponentially weighted $\ell_2$ norm. Rather than using a sum of squared dimension difference in each time period, this distance metric weights similarity in the recent past higher. In other words, treating all else as fixed, two agents choosing the same dimensions in the last period share a higher weight similarity than two agents with matching dimensions in the first period of learning. 

Using Equation \ref{eq:distance-memories}, we can define the vector RBF-kernel analog of Equation \ref{kernel} as $K'$: 
\begin{equation}
    \mathcal{K}'(\mathcal{M}_i^t, \mathcal{M}_j^t) = \exp\left(-\alpha \cdot \mathcal{D}(\mathcal{M}_i^t, \mathcal{M}_j^t)\right) \label{eq:kernelvec}
\end{equation}
and the corresponding weight matrix is determined the from $K'$ as above: 
$$  \mathcal{W}_{ij}^t =\frac{\mathcal{K}'(\mathcal{M}_i^t,\mathcal{M}_j^t)}{\sum\limits_{k=1}^n \mathcal{K}'(\mathcal{M}_i^t,\mathcal{M}_k^t)}$$

We assume that at time $t+1$, all agents can see the endogenous network structure at time $t$. So, using Theorem \ref{dimchoice:thm}, their choice of dimension in period $t+1$ is exactly the dimension $d_i$ that best distributes influence conditional on the information seen from the previous period. 

Even though the actual network updates and learning process follows DeGroot updating, the choices of dimensions to learn about in each period follows from Bayesian updating. Each agent had chosen a dimension in period $t$ based on information up until period $t-1$. They then all see the network at time $t$, and update their belief accordingly using Bayes rule. 

\begin{corollary}[Bayesian Updating Under Iterative DeGroot Updating]
    \label{corollary:iterative-multidim-DG}
    Consider a decision period $t+1$ where all agents observe the network structure from period $t$. Assume agents' beliefs follow a martingale process on past expectations: they have beliefs $\mu^t$ on the dimension choices of other agents in period $t+1$, where:
    $$\mathrm{E}_{\mu}[d^{t+1} \mid \mathcal{I}_t] = \mathrm{E}_\mu[d^t] $$ where $\mathcal{I}_t$ includes all information the agent has after observing the network at time $t$. 
    Then, using Theorem \ref{dimchoice:thm}, each agent chooses the dimension that minimizes the squared sum of stationary influences conditional on their beliefs of other agents' future dimension choices:
    $$d_i^{t+1} \in \arg \min _d \sum\limits_{k=1}^n \left(\pi_k(d, \mathrm{E}_\mu[d_{-i}^{t+1}])\right)^2$$
\end{corollary}

Details are provided in \ref{deriv:iterativeDG}. When making a decision in period $t+1$, each agent first updates their posterior over possible memory profiles. They then form an expectation of the most recent dimension choice chosen by each agent in the past period and choose their subsequent dimension accordingly. 

\section{Summary and Conclusions}
\label{conclusion}

This paper presents a social learning model under the DeGroot learning rule where agents choose the precision levels of their signals and the dimension of the state they learn about. Both choices shape the underlying network structure on which social learning occurs. We first present a tractable model with a fixed network structure and single dimension where the optimal precision choice $\tau$ is sublinear in the agent's stationary influence: $\pi$: $\tau^3_i \propto \pi_i^2$. We show how this result specifies precision choices under common network structures: complete networks, core-periphery, ring, and star networks, and explicitly compare the individually versus socially optimal choices. 

Our second main contribution is allowing the network structure to be flexible and exclusively dependent on what agents choose to learn about. We propose an RBF kernel-based distance metric between agents, which then translates to a corresponding weight matrix and network structure on which learning takes place. We show that an agent's optimal dimension choice is not one which maximizes their influence in the network but rather one that best distributes influence across agents. This theorem characterizes optimal behavior when information acquisition occurs in a single period: we then consider the natural analog where the information gathering and social learning process occurs iteratively. Distances between agents are then defined as a vector analog of the single-iteration case, and the dimension choices follow directly from Bayesian updating on other agents' future choices.  

We discuss a few interesting future directions and extensions of the paper. In ongoing research, we hope to extend the model along various directions. Our current model specifies an endogenous network structure where connections are more likely to be made between similar learning agents. However, having connections between opposite learning agents may be more beneficial to the overall network, leading to quicker convergence as information flows faster. Agents who learn about completely opposite dimensions and then interact with one another will extract the maximum amount of information from their two private signals, whereas two agents who learn about similar dimensions and then interact may spark information confounding \citep{dasaratha2019aggregative}. Agents could have the opportunity to pay to alter the network structure in their favor. In particular, an agent may pay some fixed cost to connect to an agent with a completely different learning trajectory. This explicit addition of diversity in opinion can potentially improve the speed of learning. 

Another interesting direction would be considering a sequential analog of our model. In this paper, we focus on the case in which agents all simultaneously choose a dimension, and then the network structure is endogenously formed from those choices. One could consider a case in which agents report their dimension choices one by one, and thus agents have incentives not only to best respond to past agents but also to shape the network by influencing future agent decisions. Agents may want to help the overall network by learning about an element about which fewer people have learned about, which may be suboptimal in the short run but better in the long run.

Our model can also be extended to augment the growing literature on generative AI by considering a layer of AI agents in the network. By making the status of nodes uncertain, agents are unsure of whether their neighbors are other human agents or AI agents. Agents would trust neighbors differently depending on their inherent type, and this additional layer of uncertainty would affect agent decisions and precision choices. For example, if agents perceive AI to be very precise, knowing that an AI agent has a stationary influence in the network would imply a lower choice of precision as the existing level of variance would be scaled down. This updated model structure could be used to capture the growing reliance on generative AI in human decisions. 

\newpage
\bibliography{bibliography.bib}

\appendix

\clearpage
\section{Derivations and Proofs}
\label{deriv:app}

\subsection{Theorem \ref{sublinear}}
\label{thm1_DERIV}

Our initial optimization problem is as follows: 
$$\min_{\tau \ge0} \; \mathbb{E}[(\hat{\theta}-\theta)^2] + c_i(\tau_i)$$

We can decompose the first mean squared error term in the objective function above. Since each signal distribution is consistent ($E[s_i] = \theta \; \forall i$), the overall consensus $\hat{\theta}$ is also unbiased. $\hat{\theta}$ is just a convex combination of the different agent's signals, and if each of them is consistent, the convex combination is as well. Therefore, we have that:
$$\mathbb{E}[(\hat{\theta}-\theta)^2] = \mathbb{E}[\hat{\theta}-\theta]^2 + \mathrm{Var}[\hat{\theta}-\theta]  = \mathrm{Var}[\hat{\theta}]$$

Furthermore, since signals are independent, the variance can be expressed as the sum of the weighted variances of each signal: 
$$\mathrm{Var}[\hat{\theta}] = \mathrm{Var}\left(\sum_{k=1}^n \pi_k s_k\right) = \sum_{k=1}^n \mathrm{Var}(\pi_k s_k) = \sum_{k=1}^n \pi_k^2 \cdot  \mathrm{Var}(s_k) $$
where the second step follows from the independence of signals. Under normal signals, we can further simplify $\mathrm{Var}[\hat{\theta}] = \sum_{k=1}^n \frac{\pi_k^2}{\tau_k^2}$. 

Plugging this back into our loss function gives us the following optimization function for each agent $i$:
\begin{equation*}
    \min_{\tau_i \ge0} \; \sum_{k=1}^n \pi_k^2 \cdot  \mathrm{Var}(s_k)  + c_i(\tau_i) \label{eq1}
\end{equation*}

Under regularity conditions on the network, adding more people to the network decreases consensus variance: there are additional terms in the first summation but since $\sum_k \pi_k = 1$, adding more people would lead to lower absolute influence by each agent, and thus smaller contributing variance terms by every agent. However, this does not always hold: consider an $n$ person ring network topology, where each agent's influence is effectively $\frac{1}{n}$. We could then add a central agent who is connected to all other agents; thus this new agent's influence is much larger than $\frac{1}{n}$ and the overall sum will increase. 

Theorem \ref{sublinear} then follows trivially from taking first order conditions of above with respect to $\tau_i$:
$$\pi_i^2 \cdot \frac{\partial}{\partial \tau_i}\mathrm{Var}(s_i) = c_i'(\tau_i)$$
Under normal signals, the expression simplifies to an analytic expression for 
$$\frac{2 \pi_i^2}{\tau_i^3} = c_i'(\tau_i) \Rightarrow \tau_i^3 = \frac{2 \pi_i^2}{c_i'(\tau_i)}$$ as desired.
Second-order derivatives are always positive under normal signals, and they are also positive in the general setting as long as $\frac{\partial^2}{\partial \tau_i^2}\mathrm{Var}(s_i)$ is not sufficiently negative. 

\subsection{Corollary \ref{corollary:socialplanner}}
\label{pf:socialplanner}

The social planner solves the joint optimization problem presented in Equation \ref{socialplanner}. The $n$ first order conditions are decoupled, each of which yield:
    $$\frac{\partial}{\partial \tau_i}\left(n \cdot \sum_{k=1}^n \frac{\pi_k^2}{\tau_k^2} + \sum_i c_i(\tau_i)  \right) = \frac{\partial}{\partial \tau_i} \left(n \cdot \frac{\pi_i^2}{\tau_i^2} + c_i(\tau_i)\right)$$

    $$\Rightarrow\frac{2 n\pi_i^2}{\tau_i^3} = c_i'(\tau_i) \Rightarrow \tau_i^{social} = \left(\frac{2 n\pi_i^2}{c_i'(\tau_i^{social})}\right)^{1/3}$$
    and the initial result follows. 

To show the bounds on the ratio of the individually and socially optimal precision choices, divide the social planner FOCs and individual FOCs to get: 
$$\frac{\tau_i^{social}}{\tau_i} = \biggl(\frac{n c'(\tau_i)}{c'(\tau_i^{social})}\biggr)^{1/3}$$
Let the ratio of the socially optimal to individually optimal precision choice for agent $i$ be denoted as $r_i$. 
$$\Rightarrow r_i = \biggl(\frac{n c'(\tau_i)}{c'(r_i \tau_i)}\biggr)^{1/3}$$

Assume for contradiction that $r \leq 1$. The LHS is less than 1, and since costs are convex, the ratio of the cost function derivatives must be $\geq 1$, and since $n \geq 2$, the RHS is greater than 1. $\perp$

Since costs are convex, the ratio of derivatives above must be $\leq 1$. Take the boundary case of when the ratio is equal to 1. This exactly yields linear costs, and simplifying gives that the ratio $r_i = n^{1/3}$, giving the upper bound. 

If we assume that the cost function is of the form $c_i(\tau_i) = \tau_i^\alpha$ for $\alpha \geq 1$ (to ensure convexity), then the ratio can be explicitly solved in terms of $\alpha, n$. 

$$\frac{c'(\tau_i)}{c'(r_i \tau_i)} = \frac{1}{r_i^{\alpha-1}}\Rightarrow r_i = \biggl(\frac{n}{r_i^{\alpha-1}}\biggr)^{1/3} \Rightarrow r_i = n^{\frac{1}{2 + \alpha}}$$

Under linear costs, we get the desired $n^{1/3}$ bound, but the steeper the cost function, the smaller inefficiency becomes. 

\subsection{Complete Network Structure with n Agents}
\label{proof:claimopt_prec}

To derive optimal precisions, we can use the general weight matrix and construct a corresponding system of $n$ linear equations. We have the weight matrix W as: 
$$
      W \;=\;
      \begin{pmatrix}
        x_1               & \tfrac{1 - x_1}{\,n-1\,} & \cdots & \tfrac{1 - x_1}{\,n-1\,} \\
        \tfrac{1 - x_2}{\,n-1\,} & x_2               & \cdots & \tfrac{1 - x_2}{\,n-1\,} \\
        \vdots            & \vdots            & \ddots & \vdots \\
        \tfrac{1 - x_n}{\,n-1\,} & \tfrac{1 - x_n}{\,n-1\,} & \cdots & x_n
      \end{pmatrix}
$$
and so the general expression for the $j$'th element of the stationary distribution $\pi_j$ can be written as: 
$$\pi_j = \pi_j x_j + \sum_{i \neq j} \pi_i \frac{1-x_i}{n-1}$$
Simplifying and reorganizing the expression, we get: 
\begin{align*}
    \pi_j (1-x_j) &= \frac{1}{n-1} \sum_{i \neq j} \pi_i(1- x_i) \\
    & = \frac{1}{n-1} \left(\sum\limits_{i=1}^n \pi_i(1-x_i) - \pi_j(1-x_j)\right)
\end{align*}
We then let $S = \sum\limits_{i=1}^n \pi_i(1-x_i)$, which yields: 
\begin{align*}
    & \pi_j (1-x_j) = \frac{1}{n-1}\left(S - \pi_j(1-x_j)\right) \\
    & \Rightarrow (n-1)\left[ \pi_j (1-x_j)\right] = S - \pi_j(1-x_j) \\ 
    & \Rightarrow S = n \pi_j (1-x_j) \\
    & \Rightarrow \pi_j = \frac{S}{n(1-x_j)}
\end{align*}

We use the fact that $\pi$ must be a stationary distribution and thus all elements must sum up to 1: $\sum\limits_k \pi_k = 1$. This yields: 
\begin{align*}
    &\sum\limits_{k=1}^n \frac{S}{n(1-x_k)} = 1 \\
    & \Rightarrow S = \frac{n}{\sum\limits_{k=1}^n \frac{1}{1-x_k}}
\end{align*}
Plugging this back into our expression for $\pi_j$, we have: 
$$\pi_j = \frac{\displaystyle \frac{n}{\sum_{k=1}^n \frac{1}{1 - x_k}}}{n\,(1 - x_j)} = \frac{\frac{1}{1-x_j}}{\sum\limits_{k=1}^n \frac{1}{1-x_k}}$$

The optimal precision $\tau_j$ is then directly characterized by Theorem \ref{sublinear}, where we plug in the expression above for $\pi_i$ into $\left(\frac{2\pi_i^2}{c_i'(\tau_i)}\right)^\frac{1}{3}$. 

Finally, to show that the stationary distribution is indeed strictly increasing in $x_i$, we can show that $\frac{\partial \pi_i}{\partial x_i} > 0$. Let $D = \sum\limits_{k=1}^n \frac{1}{1-x_k}$. Then: 
\begin{align*}
    \frac{\partial \pi_i}{\partial x_i} &= \frac{\frac{1}{(1-x_i)^2} \cdot D - \left(\frac{1}{1-x_i}\right)\left(\frac{1}{(1-x_i)^2}\right)}{D^2} \\
    & = \frac{D - \frac{1}{1-x_i}}{D^2(1-x_i)^2} > 0
\end{align*}
The denominator is clearly greater than 0 as both terms are squared, and the numerator is also greater than 0 due to our definition of D and the fact that $x_i \in (0, 1)$. 

\subsubsection{Numerical Examples}

The stationary distribution is $\pi = \left(\frac{1}{4},\frac{1}{4}, \frac{1}{4}, \frac{1}{4} \right)$. Using Theorem \ref{sublinear}, each agent will choose a precision level of: 
$$\tau_i^3 = \frac{2 \pi_i^2}{c_i'(\tau_i)} = \frac{1}{8\kappa} = \frac{1}{32} \Rightarrow \tau_i = 32^{-1/3} \approx 0.315$$

This will lead to each agent receiving a signal from $\mathcal{N}\left(\theta, 32^{\frac{2}{3}} \right)$.

We can then compare this to the social planner case, where agents no longer individually choose their precision levels and the social planner chooses all $\tau_i$'s. Following Corollary \ref{corollary:socialplanner}, each agent's precision choice will just be the individually optimal choice scaled by a factor of $n^\frac{1}{3} = 4^\frac{1}{3} \approx 1.6 \Rightarrow \tau_i^{social}  \approx 0.504 $. By putting this higher level of effort, the network consensus variance is reduced, whereas individual costs do end up increasing. 

We can see this by evaluating the objective at both precision choices: under the individually optimal precision choice, the objective function evaluates to: 
$$\frac{1}{4 \tau_i^2} + 4 \tau_i \approx 3.78$$
whereas under the socially optimal precision choice, the objective evaluates to:
$$\frac{1}{4 \tau_i^2} \cdot \frac{1}{1.6^2} + 6.4 \tau_i \approx 3$$

which means that the socially optimal case is better for all agents, but each agent still has an incentive to deviate and lower their precision choice: in this socially optimal case, where all other agents $i' \neq i$ follow the social planner, agent $i$ has an incentive to deviate to the individually optimal $\tau_i$:
$$\frac{3}{16\tau_i^2} \cdot \frac{1}{1.6^2} + \frac{1}{16\tau_i^2} + 4 \tau_i\approx 2.628 < 3$$

A more interesting case is a complete network but agents don't share equal/symmetric weights. In particular, consider the case in which higher indexed agents weight themselves more, and thus have a stronger influence in the network: this is presented below in Figure \ref{fig:uneven-fully-connected-square}. 

\begin{figure}[h!]
  \centering
  \begin{minipage}{0.4\textwidth}
    \[
      W = 
      \begin{pmatrix}
        0.1 & 0.3 & 0.3 & 0.3\\
        0.25 & 0.25 & 0.25 & 0.25\\
        0.2 & 0.2 & 0.4 & 0.2\\
        0.1 & 0.1 & 0.1 & 0.7
      \end{pmatrix}
    \]
  \end{minipage}
  \quad
  \begin{minipage}{0.5\textwidth}
    \centering
    \begin{tikzpicture}[
        node/.style={circle, draw, minimum size=8mm, fill=#1!30},
        interedge/.style={<->, thick},
        loopedge/.style={->, thick},
        lab/.style={fill=white, inner sep=1pt},
      ]
      \node[node=blue] (1) at (0,2)   {1};
      \node[node=blue] (2) at (2,2)   {2};
      \node[node=blue] (3) at (0,0)   {3};
      \node[node=blue] (4) at (2,0)   {4};

      \foreach \u/\v in {1/2,1/3,1/4,2/3,2/4,3/4} {
        \draw[interedge] (\u) -- (\v);
      }

      \draw[loopedge] (1) to[loop above] node[lab] {0.1} (1);
      \draw[loopedge] (2) to[loop above] node[lab] {0.25} (2);
      \draw[loopedge] (3) to[loop below] node[lab] {0.4} (3);
      \draw[loopedge] (4) to[loop below] node[lab] {0.7} (4);
    \end{tikzpicture}
  \end{minipage}
  \caption{A fully connected network of 4 agents with varied self‐loop weights and the corresponding weight matrix.}
  \label{fig:uneven-fully-connected-square}
\end{figure}

Under this network specification, the stationary distribution is $\pi = (0.149, 0.179, 0.224, 0.448)$, where the agents with higher weights on their own opinion have a stronger influence in the network (less affected by their neighbors). This implies that agent 4 will put in the most effort and thus choose the highest precision out of all agents. 

\subsection{Core-Periphery and Ring Network Structures}
\label{proof:core-perip-claim}

We focus on the core-periphery network derivation; the ring network details are analogous. We follow a similar procedure to Appendix \ref{proof:claimopt_prec}, first constructing the general stationary distribution system of equations. For the periphery agents, the general equation takes the form: 
$$\pi_j = \frac14 \pi_{j-1} + \frac14 \pi_j + \frac14\pi_{j+1} + \frac1n \pi_n$$
and for the core agent:
$$\pi_n = \sum\limits_{i=1}^{n-1} \frac14 \pi_i + \frac1n \pi_n$$

Using the fact that the stationary distribution sums up to 1 and all periphery agents are symmetric, we can solve above to get 
$$\pi_p  = \frac{4}{5n-4} \text{ and } \pi_n = \frac{n}{5n-4}$$
where $\pi_p, \pi_n$ denote stationary influence of a periphery agent and the core agent respectively. 

Logically, as the number of agents $n$ grows, both the core and periphery agents lose influence in the network. However, periphery firms lose more influence from adding an additional agent as compared to the core agent: $0 > \dfrac{\partial \pi_n}{\partial n}  = \dfrac{-4}{(5n-4)^2}> \dfrac{\partial \pi_p}{\partial n} = \dfrac{-20}{(5n-4)^2}$. 

\subsection{Corollary \ref{corollary:iterative-multidim-DG}}
\label{deriv:iterativeDG}

When moving from period $t$ to $t+1$, each agent $i$ observes the full weight matrix $W_t$ but not others’ raw dimension choices $d_j^\tau \;\;\forall \tau  \in \{1, ..., t\}$. Each agent constructs a set consistent memory vectors $\mathcal{M}(W_t)$. 
$$\mathcal{M}'(W_t) = \{ M^t = (M_1^t, M_2^t, ... M_n^t) : M^t \text{ is consistent with observed } W_t\}$$
where consistency means that applying the kernel $K'$ to agent $i$'s own known memory and all other agents' memories matches with the observed $W_{ij}$:
$$\mu_i(M^t \mid W^t) \propto \mu_i(M^t) \times \mathbbm{1}\Bigl\{W_{ij} = \tfrac{K'(M_i^t,M_j^t)}{\sum\limits_{k=1}^n K'(M_i^t,M_k^t)} \forall j\Bigr\}$$

Thus, $\mu(M^t \mid  W^t) = 0 $ if $M^t \notin \mathcal{M}'(W_t)$, $\mu(M^t \mid  W^t) \propto \mu(M^t) $ if $M^t \in \mathcal{M}'(W_t)$. 

Given this set of consistent memory vectors $\mathcal{M}'(W^t)$, the agent then forms an expectation on the dimension each other agent will choose in the next period. In particular, this expectation is just the weighted average of dimensions chosen in the past period. 
$$\mathrm{E}\bigl[d_j^{\,t+1}\mid W^t\bigr] = \sum\limits_{k=1}^m k \cdot  \Pr\bigl(d_j^t = k \mid W^t\bigr) = \sum\limits_{k=1}^m k \left(\sum\limits_{M^t \in \mathcal{M}(W^t)} \mathbbm{1} \bigl\{ d_j^t(M^t) = k  \cdot \mu(M^t \mid W^t)\bigr\} \right)$$

Then, using a direct application of Theorem \ref{dimchoice:thm} where the choices of other agents $d_{-i}$ is now constructed using this expectation, agent $i$ picks 

$$d_i^{\,t+1}\in \arg\min_{d\in\{1,\dots,m\}} \sum\limits_{k=1}^n \bigl[\pi_k\bigl(d,\; \mathrm{E}_\mu\bigl[d_{-i}^{t+1}\bigr]\bigr)\bigr]^2$$

\subsection{Welfare Analysis and Comparative Statics}
\label{pf:welfareanalysis}

We start with the individual welfare gap $\delta_i$ defined the difference between the welfare contribution of agent $i$ under the individual equilibrium and social equilibrium as follows: 

$$\delta_i = n \frac{\pi_i^2}{\tau_i^2} + c_i(\tau_i) - \Biggl(n \frac{\pi_i^2}{(\tau_i^{social})^2} + c_i(\tau_i^{social}) \Biggr) $$

Note that the socially optimal choice by definition maximizes social welfare. However, since our analysis is conducted with each agent's loss function $\ell_i(\cdot) = -u_i(\cdot)$, we flip the difference and show this gap is $\geq 0$. 

As in the proof for Corollary \ref{corollary:socialplanner}, define $r_i = \frac{\tau_i^{social}}{\tau_i}$. Rearranging the first-order conditions and using the ratio of the FOCs across the individual and social equilibria gives us: 

$$\delta_i = \frac{n\tau_ic_i'(\tau_i)}{2} + c_i(\tau_i) -\Bigg( \frac{\tau_i^{social}c_i'(\tau_i^{social})}{2}  - c_i(\tau_i^{social})\Biggr)$$
$$\Rightarrow \delta_i = \frac{n\tau_ic_i'(\tau_i)}{2} \Biggl[ 1 - \frac{1}{r_i^2}\Biggr] + c_i(\tau_i) - c_i(\tau_i^{social})$$

We can now use the convexity of our cost function to upper bound the cost difference term. In particular, the general convexity definition gives $c_i(\tau_i^{social}) - c_i(\tau_i) \geq c_i(\tau_i) (\tau_i^{social} - \tau_i)$, but simply rearranging terms gives us the more convenient lower bound: $c_i(\tau_i) - c_i(\tau_i^{social}) \leq c_i'(\tau_i) (\tau_i - \tau_i^{social})$, and this gives: 
$$ \leq \frac{n\tau_ic_i'(\tau_i)}{2} \Biggl[ 1 - \frac{1}{r_i^2}\Biggr] + c_i'(\tau_i) (\tau_i - \tau_i^{social})$$
$$ = \frac{n\tau_ic_i'(\tau_i)}{2} \Biggl[ 1 - \frac{1}{r_i^2}\Biggr] +  \tau_ic_i'(\tau_i)(1-r_i) = \tau_i c_i'(\tau_i) \Biggl[\frac{n}{2}\biggl(1-\frac{1}{r_i^2}\biggr) + 1 - r_i \Biggr]$$

We can also lower bound the difference above using an analogous application of convexity and show both bounds are of order n:
$$\geq \frac{n\tau_ic_i'(\tau_i)}{2} \Biggl[ 1 - \frac{1}{r_i^2}\Biggr] + c_i'(\tau_i^{social}) (\tau_i - \tau_i^{social})$$
From our FOCs, we know that $c_i'(\tau_i^{social})$ can be expressed in terms of $c_i'(\tau_i)$: $c_i'(\tau_i^{social}) = \frac{n}{r_i^3}c_i'(\tau_i)$ and so above simplifies further to:
$$\frac{n\tau_ic_i'(\tau_i)}{2} \Biggl[ 1 - \frac{1}{r_i^2}\Biggr] + \frac{n}{r_i^3}c_i'(\tau_i) (\tau_i - \tau_i^{social}) = \frac{n\tau_ic_i'(\tau_i)}{2} \Biggl[ 1 - \frac{1}{r_i^2}\Biggr] + \frac{n}{r_i^3}\tau_i c_i'(\tau_i)(1-r_i) $$

$$\tau_i c'_i(\tau_i) \Biggl[ \frac{n}{2} \biggl(1-\frac{1}{r_i^2}\biggr) + \frac{n}{r_i^3}(1-r_i)\Biggr] = \tau_i c'_i(\tau_i) \Biggl[\frac{n}{2} + \frac{n}{r_i^3} - \frac{2n}{r_i^2} \Biggr]$$

From Corollary \ref{corollary:socialplanner}, the ratio $r_i \leq n^{1/3}$, and since $\tau_i, c_i'(\tau_i)$ are independent of n, we see that asymptotically both terms are of order $\Theta(n)$ for each agent. Summing up over all agents gives a gap of size $\Theta(n^2)$. 

\section{Examples}
\label{examples:app}

\subsection{Specialists vs Generalists Comparison}
\label{examples_spec-vs-gen:app}

We assume that the network consists of only specialists and generalists. Specialists uniformly choose a dimension to learn about and allocate their full precision budget to. 

Thus, for some dimension $d_i$, $\frac{\alpha n}{m}$ agents have a precision of $\frac{1}{\tau_i^2}$ and the remaining $\frac{\alpha n (m-1)}{m}$ specialists have precision $\frac{1}{\underline{\tau}^2}$. All generalists have precision of $\frac{1}{\left(\underline{\tau} + \frac{\tau_i}{m}\right)^2}$. 

The variance of the network consensus on dimension $d_i$ is:
$$\mathrm{Var}[\hat{\theta}_{d_i}] = \mathrm{Var}\left(\sum\limits_{k=1}^n \pi_k s_k\right) = \pi_i^2 \; \mathrm{Var}\left(\sum\limits_{k=1}^n  s_k\right)$$ 

$$\mathrm{Var}\left(\sum\limits_{k=1}^n  s_k\right) = \left( \frac{\alpha n }{m} \cdot \frac{1}{\tau_i^2} + \frac{\alpha n (m-1)}{m} \cdot \frac{1}{\underline{\tau}^2} +  (1-\alpha)n\cdot \frac{1}{\left(\underline{\tau} + \frac{\tau_i}{m}\right)^2}\right)$$

$$\Rightarrow \mathrm{Var}[\hat{\theta}_{d_i}] = \pi_i^2 \left( \frac{\alpha n }{m} \cdot \frac{1}{\tau_i^2} + \frac{\alpha n (m-1)}{m} \cdot \frac{1}{\underline{\tau}^2} +  (1-\alpha)n\cdot \frac{1}{\left(\underline{\tau} + \frac{\tau_i}{m}\right)^2}\right)$$

From a socially optimal perspective, the social planner would want to choose the $\alpha$ that minimizes this expression. The expression is affine in $\alpha$, so the minimizing choice is a boundary case (either $\alpha = 0$ or $\alpha = 1$). Specifically, the FOC yields:

$$\pi_i^2 \cdot n\left( \frac{1}{m \tau_i^2} + \frac{(m-1)}{m \underline{\tau}^2} - \frac{1}{\left(\underline{\tau} + \frac{\tau_i}{m}\right)^2}\right) = 0$$

\noindent If the expression is greater than 0, then the derivative is increasing in $\alpha$ and so $\alpha^* = 0$. 

Note that when $\underline{\tau} << \tau_i$, then the LHS simplifies to:
$$\approx \frac{1}{m \tau_i^2} + \frac{1}{\underline{\tau}^2} - \frac{m^2}{\tau_i^2} > 0$$ and so a network with all generalists is optimal when baseline signals are very imprecise. 

\end{document}